\documentclass{article}

\usepackage[utf8]{inputenc}
\usepackage{dsfont}
\usepackage{amsmath}
\usepackage{amssymb}
\usepackage{amsthm}
\usepackage{graphicx}
\usepackage{easyReview}
\usepackage{url}
\usepackage{subfig}
\usepackage[sort]{cite}
\usepackage{multicol}
\usepackage{bm}

\renewcommand{\vec}{\mathbf}
\newcommand{\set}[1]{\left\{#1\right\}}

\newcommand{\abs}[1]{\left|#1\right|}
\newcommand{\norm}[1]{\left\|#1\right\|}

\newcommand{\R}{\mathds{R}}

\newcommand{\rank}{\mathop{\text{rank}}}

\newcommand{\eps}{\varepsilon}
\renewcommand{\d}{\tilde{d}}

\renewcommand{\vec}{\mathbf}

\newtheorem{remark}{Remark}
\newtheorem{theorem}{Theorem}
\newtheorem{definition}{Definition}
\newtheorem{lemma}{Lemma}
\newtheorem{claim}{Claim}

\begin{document}

\title{Metricizing the Euclidean Space towards Desired Distance Relations in Point Clouds}
\author{Stefan Rass$^{1,2}$, Sandra K\"onig$^3$, Shahzad Ahmad$^1$, Maksim Goman$^1$
\date{\small{$^1$ LIT Secure and Correct Systems Lab, Johannes Kepler University, Linz, Austria, email: stefan.rass@jku.at, shahzad.ahmad@jku.at, maksim.goman@jku.at\\$^2$ Institute for Artificial Intelligence and Cybersecurity, Universitaet Klagenfurt, Klagenfurt, Austria, email: stefan.rass@aau.at\\$^3$ Austrian Institute of Technology, Center for Digital Safety \& Security, Vienna, Austria, email: sandra.koenig@ait.ac.at}}}
\maketitle

\begin{abstract}
	Given a set of points in the Euclidean space $\R^\ell$ with $\ell>1$, the pairwise distances between the points are determined by their spatial location and the metric $d$ that we endow $\R^\ell$ with. Hence, the distance $d(\vec x,\vec y)=\delta$ between two points is fixed by the choice of $\vec x$ and $\vec y$ and $d$. We study the related problem of fixing the value $\delta$, and the points $\vec x,\vec y$, and ask if there is a topological metric $d$ that computes the desired distance $\delta$. We demonstrate this problem to be solvable by constructing a metric to simultaneously give desired pairwise distances between up to $O(\sqrt\ell)$ many points in $\R^\ell$. In particular, these distances can be chosen independent of any ``natural'' distance between the given points, such as Euclidean or others. Towards dropping the limit on how many points (at fixed locations) we can put into desired distance from one another, we then introduce the notion of an \emph{$\eps$-semimetric} $\tilde{d}$. This function has all properties of a metric, but allows violations of the triangle inequality up to an additive error $<\eps$. With this (mild) generalization of a topological metric, we formulate our main result: for all $\eps>0$, for all $m\geq 1$, for any choice of $m$ points $\vec y_1,\ldots,\vec y_m\in\R^\ell$, and all chosen sets of values $\set{\delta_{ij}\geq 0: 1\leq i<j\leq m}$, there exists an $\eps$-semimetric $\tilde{\delta}:\R^\ell\times \R^\ell\to\R$ such that $\tilde{d}(\vec y_i,\vec y_j)=\delta_{ij}$, i.e., the desired distances are accomplished, irrespectively of the topology that the Euclidean or other norms would induce. The order of quantifiers is important here: we first can choose the accuracy $\eps$ by which our semi-metric may be violate the triangle inequality (while leaving the other metric axioms to hold as usual for $\tilde{d}$), then fix the spatial locations of points, and after that step, choose the distances that we wish between our points. We showcase our results by using them to ``attack'' unsupervised learning algorithms, specifically $k$-Means and density-based (DBSCAN) clustering algorithms. These have manifold applications in artificial intelligence, and letting them run with externally provided distance measures constructed in the way as shown here, can make clustering algorithms produce results that are pre-determined and hence malleable. This demonstrates that the results of clustering algorithms may not generally be trustworthy, unless there is a standardized and fixed prescription to use a specific distance function.
\end{abstract}

\section{Introduction}
This work is motivated by the quesion of whether it is possible to manipulate clustering algorithms in artificial intelligence. Specifically, if we do not want to touch the data (e.g., by adding noise), and also do not want to change the algorithm away from trustworthy established standard techniques, can we trick a standard (unmodified) clustering algorithm into giving us results that we like for input data that is fixed? A possible positive answer is obtained if we let the algorithm run with a modified distance function provided externally, but if we want such a setup to ``appear plausible'', the distance function should not be just arbitrary, but at least satisfy the same properties as a metric in the topological sense. This leads to the challenge stated in the abstract, asking to find a metric that puts a given fixed set of points into distances that are arbitrary and up to a free choice (within obvious constraints, like there being no negative distances, for example). Before going into details about the construction of such metrics, let us briefly review their use for clustering to further substantiate the goals of this work.

Clustering algorithms, in the simplest instance, classify an input $\vec x$ by looking for the closest point $\vec y$, according to some distance measure, whose category ($class$) is known. If $\vec y$ has category $class(\vec y)$, then $\vec x$ is assigned the same category $class(\vec y)$. For unsupervised learning, the process works similar by assigning two points in close proximity the ``same'' class label, and giving other or introducing new class labels for points at some minimum distance or farther. Most clustering algorithms follow this basic approach, and distinguish themselves primarily in deeper details, such as whether or not cluster centers or the number of clusters needs to be known in advance (e.g., $k$-Means) or not (e.g., DBSCAN), how cluster centroids are determined based on data directly, the data distribution, or the data range (e.g., such as in grid-based techniques), or whether clusters are formed from points that are nearby one another, or near some fixed location. Essentially, it all boils down to a measuring ``proximity'' of points to other points. \emph{Metrics} are natural measures of distance, being positive definite, obeying the triangle inequality, and symmetric; with the latter meaning that the distance $d(\vec x,\vec y)$ between two points $\vec x,\vec y$ is the same as the distance $d(\vec y,\vec x)$. Dropping the symmetry requirement turns $d$ into a \emph{quasimetric} (although the name may vary in the literature). Even more generally, giving up on the triangle inequality, we get towards a \emph{premetric}, sometimes also just called a \emph{dissimilarity measure}. It should be noted that also dropping the positive definiteness would take us to non-Hausdorff spaces, in which two points $\vec x\neq \vec y$ can have zero distance. This case is where we draw the line as being not immediately interesting, since a clustering algorithm could in that case not distinguish two distinct points, based on their ``measured distance'', which would be awkward for practical clustering applications.

The question studied in this work is the following: 
\begin{quote}
	Let a set of points in the Euclidean space $\R^{\ell}$ be given and be fixed. Can we endow $\R^{\ell}$ with a metric that puts the given point cloud into pairwise distances of our free choice? 
\end{quote}
In other words, can we change local neighborhoods as induced by, say, the Euclidean distance, entirely by changing the metric? For example, suppose that three points $\vec x,\vec y,\vec z\in\R^{\ell}$ are given, where $\vec x$ and $\vec y$ are very close to one another, and $\vec z$ is far distant from both, $\vec x$ and $\vec y$, all according to the Euclidean metric. Can we ``change'' the metric to some other function that may put $\vec x$ and $\vec z$ close to one another, and instead has $\vec y$ come to lie far distant from $\vec x$ and $\vec z$? Or, alternatively, can we construct another metric that puts all three $\vec x,\vec y,\vec z$ into large distance from one another (thus, enlarging the close distance from $\vec x$ to $\vec y$, but leaving the distant $\vec z$ as being remote)?

This question, as formulated, admits little hope for a positive answer, considering that metrics induced by norms are all equivalent on $\R^{\ell}$, meaning that whatever norm we endow $\R^{\ell}$ with, its induced metric will, up to positive multiplicative constants, retain closeness and large distances, as implied by the Euclidean metric.

However, it turns out that -- up to an arbitrary little relaxation of the triangle inequality -- it is indeed possible to design a distance function that is ``almost'' a metric (only different in an arbitrary small error in the triangle inequality that we can even make lower than the machine accuracy), which computes any desired set of pairwise distances between any given set of points at fixed locations in $\R^{\ell}$. Using this distance function in replacement of the Euclidean or other distance, we will show how to manipulate two standard clustering algorithms towards giving \emph{any} result that we desire.

\subsection{Showcase: Security of Clustering Algorithms}
Throughout the rest of this work, let $Y\subset\R^{\ell}$ be a finite given and fixed point cloud. Assuming that we can metricize the space $\R^{\ell}$ in a certain fashion, we can then let a clustering algorithm run on the point cloud $Y$ to assign points to distinct clusters, based on their (mutual geometric) distances. The Euclidean norm, or the Manhattan norm are two typical choices for distance functions, but conceptually, clustering algorithms may work with any distance metric.

This opens an interesting possibility to manipulate an automated clustering, since if we can place a set of points into $\R^{\ell}$ to act as our cluster centers, such that the cluster centers are close or far away from each other, but will have their neighborhoods assigned to the provided cluster center, we can drive the clustering algorithm into producing a labeling that we can choose beforehand.

Figure \ref{fig:metric-change-example} illustrates the basic idea of this attack: consider the point cloud as drawn in the plane, with the points $\vec z_1, \vec z_2, \vec z_3$ being placed as cluster centers into the area. Now, suppose that we can design a metric that puts $\vec z_1$ and $\vec z_2$ close to one another, but leaves $\vec z_3$ far from both cluster centers. The aforementioned effect of a cluster center ``attracting'' points from its proximity to have the class label of it assigend to neighboring points, then would put the surrounding points $\vec y_1,\ldots,\vec y_{10}$ into the neighborhoods of $\vec z_1,\vec z_2$ and $\vec z_3$, although the Euclidean distances draw an entirely different picture (Figure \ref{fig:metric-change-example-left}). The question of this work is can we design the metric such that it produces a ``clustering picture'' that we like (e.g., such as Figure \ref{fig:metric-change-example-right}).

\begin{figure}
	\subfloat[Proximity according to Euclidean distance]{\label{fig:metric-change-example-left}
		\includegraphics[width=0.45\textwidth]{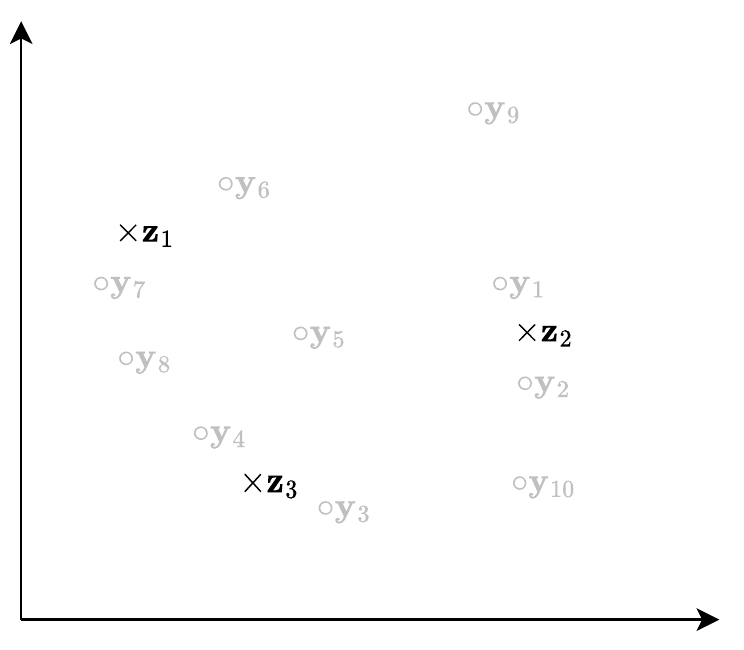}}\quad 
	\subfloat[Proximity according to some other metric]{\label{fig:metric-change-example-right}
		\includegraphics[width=0.45\textwidth]{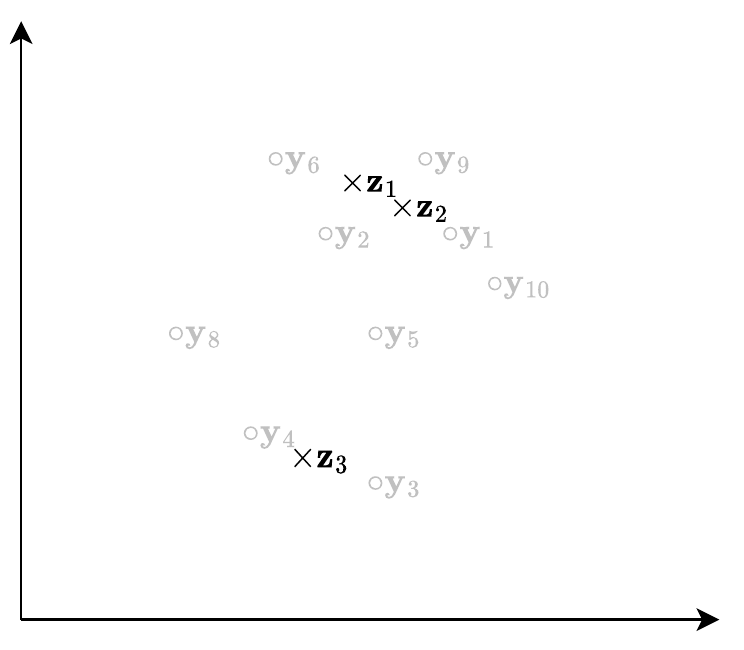}}
	\caption{Different neighborhoods depending on the underlying metric}\label{fig:metric-change-example}
\end{figure}

\subsection{Our Contribution}

We stress that the picture drawn by Figure \ref{fig:metric-change-example} is not meant to be topologically precise, nor would our results admit putting three points into the 2-dimensional space (this will be precluded for theoretical reasons in our construction), but the figure should help conveying the message about the possibility that we \emph{can} certify: 
\begin{quote}
	It is indeed possible to design metrics so that a certain number of points, irrespectively of where they are in the Euclidean space, come to lie at distances that can be chosen freely, up to satisfying the properties of a distance (symmetry and positivity for distinct points), but not necessarily the triangle inequality. 
\end{quote}

More specifically, our main results will be twofold:
\begin{itemize}
	\item We show that within $\R^{\ell}$, we can place $O(\sqrt{\ell})$ many points anywhere, and define pairwise distances between them, irrespectively of their Euclidean distance. Then, we can design a metric that puts these points exactly at distances that we have chosen before. This metric, which is a full metric in the topological sense, including the triangle inequality to hold (without error), is constructed by Theorem \ref{thm:relative-anzahl-absoluter-abstaende}.
	\item Generalizing the previous result, we then showing how we can place a larger (arbitrary) set $Z$ of points such that (i) their mutual distances, relative to each other, are as we desire, and (ii) the given data points are always close to points in $Z$. This will be Theorem \ref{thm:embedding-any-number-of-points}.
\end{itemize}

Practically, the implications of the ability to manipulate a clustering algorithm towards arbitrary results, are wide-reaching, since among many others, consider the following possibilities opened hereby: 

\begin{itemize}
	\item Unnecessary surgery, done to earn money and later claimed to be legitimate by relying on some AI.
	\item Attacking military targets according to one's own free will, but later claiming to have an AI providing the classification as the target to be worth bombing.
	\item Assign exam grades to students randomly or as we please, and fit them to a recognized scale afterwards, using AI to explain why so and pretend objectiveness hereby.
	\item Seemingly ``unbiased'' recruitment system that use AI to minimize the involvement of possibly biased humans in the selection process. Irrespectively of any debiasing that was done to the training data or the clustering algorithm itself, the possibility of changing the metric accordingly to manipulate the outcome can render all previous debiasing precautions void. The company can still hire ``as they wish'' and apparently legitimate their choice as objective.
\item Prioritization of patients based on their health conditions with help of an AI. Certain conditions (also non-medical ones) may be given priority no matter how different or similar two patients may be in medical aspects of treatment or urgency to get treated (for example, the patient with better insurance can be preferred).
\item and possibly many more, all relating to the question of whether AI can be trustworthy. Our answer is ``yes, if we a priori commit to what we do in a transparent and verifiable way''.
\end{itemize}

The contribution of this work is thus showing that independent of the data, including biases, noise or other features therein, the clustering algorithm itself is manipulable, if we admit too many degrees of freedom in the choice of the (topological) metric to run it. Thus, making AI objective and hence trustworthy requires a detailed commitment and openness not only about the algorithms themselves, but also about their configuration.

\subsection{Preliminaries}

\subsubsection*{Symbols and Notation}
The notation so far, and to be continued, lets vectors appear as lower-case bold printed letters, while scalars variables and functions are lower-case Latin letters. Greek letters are used for constants, while admitting that some constants may depend on other constants (but are hence themselves not variable). Upper case letters denote sets, and when bold-printed, denote matrices. Finally, calligraphic letters denote probability distributions. We hereafter call the set $\R$ or the higher-dimensional vector space $\R^{\ell}$ the Euclidean space (for short), \emph{not} implying herewith that it is endowed with the Euclidean metric.

We let the symbol $B_{2,\eps}=\set{\vec z\in\R^n~:~\norm{\vec z-\vec y}_2<\eps}$ be the Euclidean neighborhood of a point $\vec y\in\R^n$, being the open ball of radius $\eps$. More generally, for a norm induced by a positive definite quadratic form $q:\R^n\to\R$, we will write $B_{q,\eps}$ for the radius $\eps$-ball, w.r.t. the norm $\norm{\vec x}_q:=\sqrt{q(\vec x)}$. Since later, we will make use of random points chosen uniformly from a neighborhood, we will let the symbol $\mathcal{U}(B)$ to denote the continuous uniform distribution supported on a set $B$ of the Euclidean vector space.

In the following, we will make frequent use of distances $\delta_{ij}$ between two points $\vec y_i, \vec y_j$. We will write the distance $\delta_{ij}$ synonymously as $\delta_k$, with the convention that the integer $k$ one-to-one corresponds to the pair $ij$, written as $k\simeq ij$ in a slight abuse of notation. The reason is that we will several times need to denote an arbitrary pair or enumerate all pairs, and use a ``single'' index $k$ for this purpose, while at other times, we may speak about a specific pair, in which case we use the (equivalent) double index $ij$.

\subsubsection*{Definitions}
As the introduction already outlined, we will hereafter make use of different properties for a distance measure, and therefore, give a definition of a metric together with a tabular overview of generalizations. We stress that the terms used here are not all standardized in the literature, and some authors may prefer slight variations. 

\begin{definition}[Metric, ($\eps$-)Semimetric, and Generalizations]
Let $V$ be a vector space, and let $d:V\times V\to\R$ be a function. The function $d$ is called a \emph{metric}, \emph{semimetric}, \emph{premetric} or \emph{quasimetric}, depending on which conditions of the following list is satisfied:
\begin{description}
	\item[(I)] Identity of indiscernibles: $\forall x:d(x,x)=0$.
	\item[(P)] Positivity: $\forall x\neq y:d(x,y)>0$.
	\item[(S)] Symmetry: $\forall x,y:d(x,y)=d(y,x)$.
	\item[(T)] Triangle inequality $d(x,y)\leq d(x,z)+d(z,y)$.
\end{description}

Table \ref{tbl:metric-definitions} specifies a ``yes'' if a condition is required for the respective term, and a ``no'' if the condition is not demanded (although it may hold). 

\begin{table}[h!]
	\caption{Definition of a metric and generalizations thereof}\label{tbl:metric-definitions}
	\begin{tabular}{|l|l|l|l|l|p{3.2cm}|}
		\hline 
		 & \emph{premetric} & \emph{metric} & \emph{quasimetric} & \emph{semimetric} & \emph{$\varepsilon$-semimetric}\tabularnewline
		\hline 
		\hline 
		(I) & yes & yes & yes & yes & yes\tabularnewline
		\hline 
		(P) & yes & yes & yes & yes & yes\tabularnewline
		\hline 
		(S) & no & yes & no & yes & yes\tabularnewline
		\hline 
		(T) & no & yes & yes & no & satisfied up to an additive $\varepsilon$-error on the right\tabularnewline
		\hline 
	\end{tabular}
\end{table}

For any given $\eps>0$, we call $d$ an \emph{$\eps$-semimetric}, if it satisfies the triangle inequality up to an additive error $\leq\eps$ on the right-hand side, i.e., an $\eps$-semimetric, satisfies (I), (P), (S) and 
\begin{equation}\label{eqn:epsilon-triangle-inequality}
	\forall x,y,z\in V: d(x,y)\leq  d(x,z)+d(z,y)+\epsilon.
\end{equation} 
\end{definition}
By our definition, a $0$-semimetric is a metric, or otherwise saying that if the $\eps$-semimetric $d$ is actually ``zero-semi'', then it is a metric.

\subsection{The Problem}

Our goal is to endow the space $\R^{\ell}$ with a metric $d$, such that a set $Y=\set{\vec y_1,\ldots,\vec y_m}$ of points comes to lie in $\R^{\ell}$ at pairwise distances $d(\vec y_i,\vec y_j)=\delta_{ij}$ for all $\vec y_i,\vec y_j\in Y$, where we can choose the distance value $\delta_{ij}$ freely and in advance (but consistent to be a reasonable distance, i.e., non-negative and symmetric), i.e., we require that
\begin{enumerate}
	\item $d$ is a topological metric,
	\item $\delta_{ij}=\delta_{ji}>0$ for all $i\neq j$, but not further constrained.
\end{enumerate}

It is easy to see that finding such a $d$ will not generally be possible, since once the locations of the data points are fixed, their Euclidean distances are determined, and any other norm-induced metric would necessarily be bounded within constant multiples of these distances. However, by scaling down to local neighborhoods, we can accomplish pairwise closeness/distance relations, in the sense of, for example, putting a point $\vec y_i$ closer to $\vec y_j$ than to $\vec y_k$, although the Euclidean distances would induce different neighborship. 

Specifically, we cannot hope that our metric $d$ will satisfy exactly $d(\vec y_i,\vec y_j)=\delta_{ij}$, but we can design it to satisfy almost this equation, up to a multiplicative constant that depends on the set $Y$, i.e., we can accomplish $d(\vec y_i,\vec y_j)=\alpha\cdot\delta_{ij}$ for all $i,j$, and some value $\alpha>0$ that depends only on $Y$ (as a whole), but does not change for different $i,j$'s.

This is enough for a manipulation of a clustering: if the algorithm needs to choose a class label based on ``closeness'', the situation that $\vec y_i$ is closer to $\vec y_j$ than it is to $\vec y_k$ is reflected by the distance relation $\delta_{ij}<\delta_{ik}$. This inequality remains intact if we multiply by $\alpha>0$, so the ``closest'' cluster to $\vec y_i$ remains unchanged if we manage to put it at a distance that is proportional (not necessarily equal) to what we desire. This will be Theorem \ref{thm:embedding-any-number-of-points}.


\section{Results}



Given the set $Y\subset\R^{\ell}$ of points and having chosen the distances $\delta_{ij}$ at which we want $\vec y_i$ and $\vec y_j$ be separated for all $i<j$, we now proceed by showing how we can accomplish this as long as $m=\abs{Y}\in O(\sqrt{\ell})$ holds. Thinking about these points as a priori defined cluster centers, such as in algorithms like $k$-Means, we can place the cluster centers wherever we like and at any distance that we desire, which will cause the points in topological vicinity to be assigned to the nearest cluster center, just as we like it. 

\subsection{Embedding Points at desired Distances}\label{sec:embedding points at}


We start with a simple technical result that asserts a set of points drawn at random is almost surely linearly independent. The crucial point is that all these distributions are absolutely continuous, which we can use to prove its full rank. 
\begin{lemma}\label{lem:almost-surely-rank-1}
	Let $\vec x_1,\ldots,\vec x_n\in \R^n$ be independently, but not necessarily identically sampled from distributions that are all absolutely continuous w.r.t. the Lebesgue measure on $\R^n$. Put the vectors as columns into an $(n\times n)$-matrix $\vec M$.
	
	Then, $\vec M$ has almost surely full rank, i.e., $\Pr(\rank(\vec M)=n)=1$.
\end{lemma}
\begin{proof} We give the proof for uniform distributions, but remark that the argument only depends on the continuity of the distribution (w.r.t. the Lebesgue measure), but this more general claim is not required in the following.
	
We induct on the number $k$ of columns of the matrix $\vec M$: the case $k=1$ is trivial, so let us assume that up to $k<n$ columns $\vec x_1,\ldots,\vec x_k\in\R^n$ have been sampled linearly independent as $\vec x_1\sim F_1, \ldots, \vec x_k\sim F_k$. The new vector $\vec x_{k+1}\in\R^n$, to be a linear combination of $\vec x_1,\ldots,\vec x_k$, would need to lie in the $k$-dimensional subspace of $\R^n$. This subspace has dimension $k<n$ and is therefore a nullset w.r.t. the Lebesgue measure on $\R^n$, and since we sample from absolutely continuous distributions on $\R^n$ w.r.t. the $n$-dimensional (Lebesgue) measure, we have $\Pr_{F_{k+1}}(\vec x_{k+1}\in\text{span}\set{\vec x_1,\ldots,\vec x_k})=0$ too under the distribution $F_{k+1}$, thus completing the induction step.
\end{proof} 
Armed with Lemma \ref{lem:almost-surely-rank-1} as the tool to assure that we can select linearly independent vectors in $\R^\ell$ with high probability, we can state our first main result.

\begin{theorem}\label{thm:relative-anzahl-absoluter-abstaende}
Let $\ell>1$ and choose any set $D=\set{\vec x_1,\ldots,\vec x_n}\subset \R^{\ell}$ of arbitrary but linearly independent points. To each $\vec x_k$, associate a value $\delta_k>0$. Then, there is a norm $\norm{\cdot}_q$ that satisfies $\norm{\vec x_k}_q=\delta_k$ for all $k=1,2,\ldots,n$.
\end{theorem}
\begin{proof}
Since $\abs{D}<\infty$, we can enumerate the elements of $D$ as $D=\set{\vec x_1,\ldots,\vec x_n}$ in any fixed order (e.g., lexicographic or other). We show that we can define a quadratic form $q:\R^{\ell}\to\R$ that takes any desired value $\delta_k^2$ at point $\vec x_k$ for all $k=1,2,\ldots,n$.

For the moment, take $k$ as fixed but arbitrary, and let the matrix $\vec B_k$ be such that its null-space is the subspace of $\R^{\ell}$ spanned by $\set{\vec x_j: j\neq k}$ (such a matrix is easy to find using singular value decomposition). Define $\tilde{\vec A}_k := \vec B_k^\top\vec B_k$, then $\tilde{\vec A}_k$ is positive semi-definite. Moreover, since all $\vec x_j$ with $j\neq k$ are linearly independent of $\vec x_k$ by hypothesis, we have $\vec x_k$ falling outside the null-space of $\vec B_k$, so that $\vec x_k^\top\tilde{\vec A}_k\vec x_k>0$, and $\vec x_j^\top\tilde{\vec A}_k\vec x_j=0$ for all $j\neq k$. Finally, we will use the scaled matrix $\vec A_k:=\frac{\delta_k^2}{\vec x_k^\top\tilde{\vec A}_k\vec x_k}\cdot\tilde{\vec A}_k$ and define the quadratic form $q_k(\vec x) := \vec x^\top\vec A_k\vec x$, which now satisfies
\[
q_k(\vec x) = \begin{cases}
	\delta_k^2 & \text{ for }\vec x=\vec x_k;\\
	0 & \text{ for }\vec x=\vec x_j, j\neq k.
\end{cases}
\]

Repeating this construction for all $k=1,2,\ldots,n$, we can define a ``Lagrange-interpolation-like'' quadratic form with matrix $\vec A:=\sum_{k=1}^n \vec A_k$, which is $q(\vec x) := q_1(\vec x)+q_2(\vec x)+\ldots q_n(\vec x)$. From this, we can define the norm $\norm{\vec x}_q := \sqrt{q(\vec x)}$.

By construction, this norm achieves the desired, since for every $\vec x_k$, we get  $\sqrt{q(\vec x_k)}=\sqrt{\delta_k^2}=\delta_k$.
\end{proof}

We seek to apply Theorem \ref{thm:relative-anzahl-absoluter-abstaende} to points $\vec x_k$ that arise as pairwise differences between a set of given data points $\vec y_1,\ldots, \vec y_m$ whose classification we want to forge. However, we cannot directly define $\vec x_k:=\vec y_i-\vec y_j$ as a sequence of pairwise differences, since these will (unless $m=2$) never be linearly independet. Moreover, the number of pairwise differences between $m$ data points in total is $\binom m 2$, which needs to be $\leq \ell$ necessarily, as any larger set cannot be linearly independent within $\R^\ell$. This puts a limit of at most 
\begin{equation}\label{eqn:embeddable-points}
	m\leq \frac 1 2(1+\sqrt{8\ell+1})
\end{equation}
data points that we can work with. Even under this limit, we will still need to add random distortions to the points for linear independence, but Lemma \ref{lem:almost-surely-rank-1} tells that the setting $\vec x_k := \vec y_i-\vec y_j+\eps_{ij}$ for random noise $\eps_{ij}$ (with absolutely continuous probability distribution) will almost surely give the desired independence, for all $1\leq i,j\leq m$ and hence $k=1,2,\ldots,\binom m 2$. Under this setting, we arrive \emph{almost} at a metric, since $\norm{\vec x_k}_q=\norm{\vec y_i-\vec y_j+\eps_{ij}}_q$ differs from a metric only in the additive $\eps_{ij}$-error term. Using the triangle inequality, we see that $\norm{\vec y_i-\vec y_j+\eps_{ij}}_q\leq \norm{\vec y_i-\vec y_j}_q+\norm{\eps_{ij}}_q=d_q(\vec y_i,\vec y_j)+\norm{\eps_{ij}}_q$, where $d_q$ is the metric that the norm $\norm{\cdot}_q$ canonically induces. The additive error is what comes in new, hence pointing already out that we may need to allow some limited violation of the triangle inequality (motivating our introduction of an $\eps$-semimetric). We will make this intuition rigorous in the next section.

Informally, for a maximum number of points as given by \eqref{eqn:embeddable-points} and with random noise added, Theorem \ref{thm:relative-anzahl-absoluter-abstaende} supports the following statement at least: 
\begin{quote}
	Inside $\R^{\ell}$, we can fix any set of $O(\sqrt \ell)$ many points almost anywhere (up to random displacement for linear independence of the differences), and define a metric that puts these points into a distances that we can choose irrespectively of the location of the points.
\end{quote}

Theorem \ref{thm:relative-anzahl-absoluter-abstaende} thus lets us put a ``relative'' number of points (relative to the dimension of the space) at absolute (chosen arbitrarily) distances. We now show how to twist this around into the possibility of placing an absolute (arbitrarily chosen) number of points such that we can still find a norm that gives us distances relative, in the sense of proportional, to the pre-defined distances.


\subsection{Dropping the constraint on $m$}\label{sec:generalized-construction}
To escape the dimensionality bound that limits the number of points whose pairwise distance vectors we can have as linearly independent, we transfer the construction into a larger dimensional space, and will map each point to a higher-dimensional pendant that serves to compute the distance to all neighbors. The construction is made explicit in the proof of the following result:

\begin{theorem}\label{thm:embedding-any-number-of-points}
Let a set of points $\vec y_1,\ldots,\vec y_m\in\R^{\ell}$ with $m>1$ and $\ell>1$ be given and let $h=\binom{m}{2}$. To each pair of distinct points $\vec y_i,\vec y_j$, assign a positive number $\delta_{ij}$, subject to the constraint that $\delta_{ij}=\delta_{ji}>0$ for $i\neq j$. Fix an $\eps>0$ such that $B_{2,\eps}(\vec y_i)\cap B_{2,\eps}(\vec y_j)=\emptyset$ for all $i\neq j$.

Then, with probability 1, we can endow $\R^h$ with a norm $\norm{\cdot}_Q$ with the following properties: 
\begin{enumerate}
	\item[(a)] for every triple $i,j,k$, there are points $\vec z_i,\vec z_j,\vec z_k$ in a neighborhood
	
	$\norm{\vec y_i-\vec z_i}_2<\eps, \norm{\vec y_j-\vec z_j}_2<\eps$ and $\norm{\vec y_k-\vec z_k}_2<\eps$, such that the distance relation $\delta_{ij}\lesseqqgtr \delta_{ik}$ holds if and only if $\norm{\vec z_i-\vec z_j}_Q\lesseqqgtr\norm{\vec z_i-\vec z_k}_Q$. The symbol $\lesseqqgtr$ will practically become an explicit $<$, $=$ or $>$, depending on the left- and right-sides.
	\item[(b)] $B_{Q,\eps}(\vec y):=\set{\vec z\in\R^h~:~\norm{\vec y-\vec z}_Q<\eps}\subseteq B_{2,\eps}(\vec y)$ for all $\vec y\in\set{\vec y_1,\ldots,\vec y_m}$, i.e., disjoint Euclidean neighborhoods remain disjoint in the topology that the norm $\norm{\cdot}_Q$ induces.
\end{enumerate}
\end{theorem}

The neighborhood of each $\vec z_i$ will contain the given data point $\vec y_i$ (also $\vec y_i'$), and be in the \emph{desired} distance \emph{relation} to all other points by part (a) of the above theorem. This is all we need for a clustering, based on proximity relations, to come up arbitrarily as we wish.

\begin{proof}[Proof of Theorem \ref{thm:embedding-any-number-of-points}]
We will place a ``virtual neighbor'' near each data point $\vec y_1,\ldots,\vec y_m$, such that the distances between the virtual neighbors of each data points are freely specifiable. To bypass the $O(\sqrt \ell)$-bound on $m$ given by \eqref{eqn:embeddable-points}, we will respectively increase the dimension towards $h=\binom m 2$ to have a sufficient number of pairs embeddable with linearly independent differences, and in close proximity to the given data points $\vec y_1,\ldots,\vec y_m$. 

Before materializing this intuition, let us be clear on the symbols to be used hereafter, to ease following the argument. We will let:
\begin{itemize}
	\item $\vec y_i\in\R^{\ell}$ be a given data point from $Y$,
	\item $\vec y_i'\in\R^h$ be the canonic embedding of $\vec y_i\in\R^{\ell}$ inside $\R^h$.
	\item $\vec z_{i,s}\in\R^h$ be a neighbor of $\vec y_i'\in\R^h$ that is assigned a designed distance $\delta_{sr}$ to a corresponding neighbor $\vec z_{j,r}$ of the data point $\vec y_j'\in\R^h$
\end{itemize}

Let us first put these given points into the higher-dimensional space $\R^h$ as $\vec y_i'=(\vec y_i^\top\;\vec 0_{(h-\ell)\times 1})^\top\in\R^h$. Then, fix some positive $\eps<\frac 1 2\min_{i,j}\norm{\vec y_i'-\vec y_j'}_2$, and to each point $\vec y_i'$, associate a set of $m-1$ random points $\vec z_{i,1},\ldots,\vec z_{i,m-1}$ within an Euclidean distance $<\eps$ towards $\vec y_i$. These points will serve to measure the pairwise distances between $\vec y_i$ and $\vec y_j$ using higher-dimensional ``substitutes''. Note that we cannot just use $\vec y_i', \vec y_j'$ directly, since the pairwise differences between them could not give us linearly independent vectors as we require. To see this, recall that we have $m$ such vectors $\vec y_1',\ldots,\vec y_m'$, and computing pairwise differences is a linear operation. Hence, the result cannot be linearly independent vectors. For that to happen, we need points that are ``independent'' of $\vec y_i,\vec y_j$ to give linearly independent differences, but still close enough to $\vec y_i,\vec y_j$ to ``represent'' their spatial location. Therefore, we sample from an $\eps$-neighborhood that is small enough to ensure disjointness between any two points $\vec y_i', \vec y_j'$, and sample a fresh random point near $\vec y_i$ whenever we need a distance vector between $\vec y_i$ and one of the other points. Figure \ref{fig:random-neighbors} shows this.

\begin{figure}
	\centering
	\includegraphics[scale=0.7]{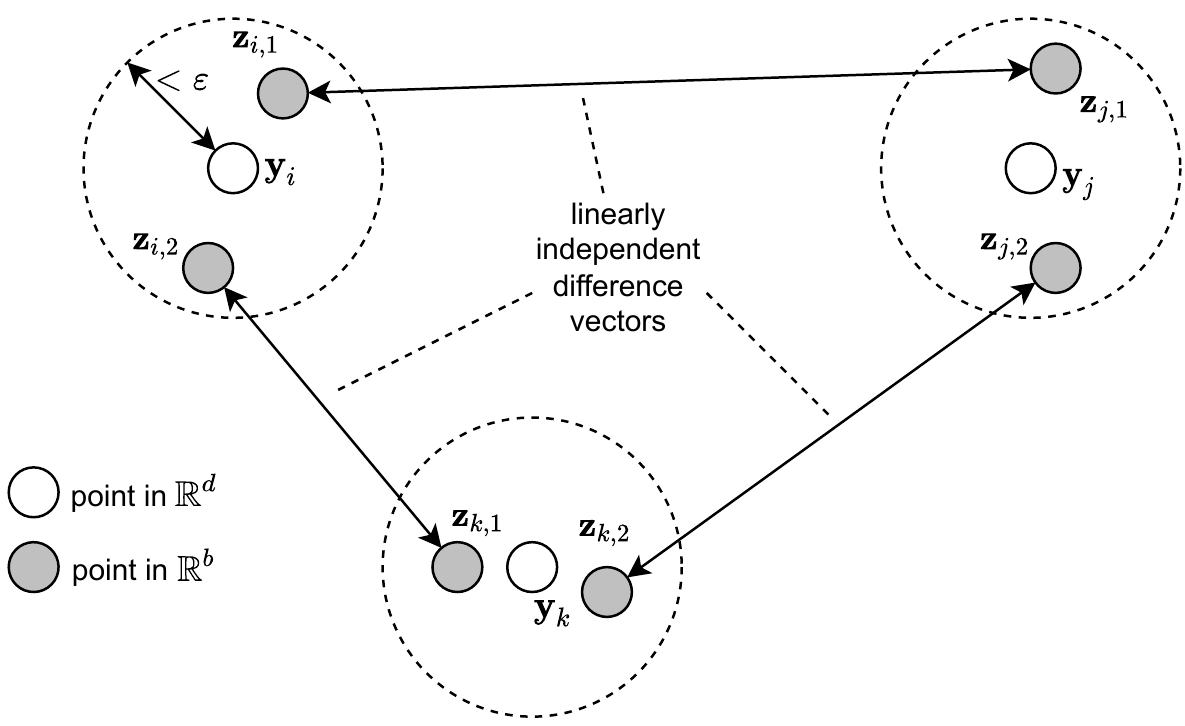}
	\caption{Linear independence by (stochastically) independent neighbor choices}\label{fig:random-neighbors}
\end{figure}

To see why this establishes the required linear independence, let us systematically iterate through the pairwise differences to see why linear independence is accomplished by the construction as described. 

Start with $\vec y_1\in\R^{\ell}$, respectively $\vec y_1'\in\R^h$: for all other points $\vec y_j'$ for $j=2,\ldots,m$, we sample a fresh neighbor $\vec z_{i,j}\in \mathcal{U}(B_{2,\eps}(\vec y_i'))$ and $\vec z_{j,1}\in \mathcal{U}(B_{2,\eps}(\vec y_j'))$ uniformly at random. We then compute the distance vector $\vec x_k=\vec z_{i,j}-\vec z_{j,1}$ (with the index $k\simeq ij$). This procedure is likewise repeated for $i=2,\ldots,m$, giving the difference vectors $\vec x_1,\vec x_2,\ldots,\vec x_{\binom m 2}\in\R^h$, all collected in a matrix
\begin{equation}\label{eqn:extended-matrix}
	\vec M=\left(\vec x_1, \ldots, \vec x_h\right)\in\R^{h\times h}.
\end{equation}

\begin{claim}\label{lem:almost-surely-linearly-independent-neighbors}
	Let $\vec y_1,\ldots,\vec y_m$ be given, 
	The matrix $\vec M$ constructed from pairwise differences of uniformly random points $\vec x_i\sim\mathcal{U}(B_{2,\eps}(\vec y_i))$ within neighborhoods centered at the points $\vec y_1',\ldots,\vec y_m'$, has almost surely full rank, i.e., $\Pr(\rank(\vec M)=h)=1$.
\end{claim}
\begin{proof}[Proof of Claim \ref{lem:almost-surely-linearly-independent-neighbors}]
	This is a mere application of Lemma \ref{lem:almost-surely-rank-1} with the distribution for $\vec x_i$ being $\mathcal{U}(B_{2,\eps}(\vec y_i))$. Since these are all absolutely continuous, the coordinates of each $\vec x_k$ are differences of independent uniformly distributed random variables, and as such have a trapezoidal shaped distribution, which is again absolutely continuous, and Lemma \ref{lem:almost-surely-linearly-independent-neighbors} applies.
\end{proof}


Claim \ref{lem:almost-surely-linearly-independent-neighbors} assures the conditions needed for the construction in Section \ref{sec:embedding points at} are met with probability 1, and it just remains to assign the distances between the substitute points $\vec z_{i,s},\vec z_{j,t}$ (for $1\leq s,t\leq m-1$) in the neighborhood of $\vec y_i$ (resp. $\vec y_i'$) and $\vec y_j$ (resp. $\vec y_j'$) to be as we desire.


Now, repeating the construction from Section \ref{sec:embedding points at} on the pairwise difference vectors (that are now all linearly independent, since they were computed from distinct $z$-neighbors of the original points), we get a quadratic form $q_h$ on the higher-dimensional space $\R^h$. The resulting norm $\norm{\vec x}_q:=\sqrt{q_h(\vec x)}$ is equivalent to the Euclidean norm, and the respective constants are given by the eigenvalues of the matrix $\vec A$ within $q_h(\vec x)=\vec x^\top\vec A\vec x$, since every quadratic form satisfies
\[
\lambda_{\min}(\vec A)\cdot\norm{\vec x}_2^2\leq \vec x^\top\vec A\vec x\leq \lambda_{\max}(\vec A)\cdot\norm{\vec x}_2^2,
\]
where $\lambda_{\min}(\vec A), \lambda_{\max}(\vec A)$ are the smallest and largest eigenvalues of $\vec A$. Experimentally, one quickly discovers that there is not much control over the eigenvalues of $\vec A$ directly, and this is not surprising, since the norm induced by $q_h$ cannot diverge arbitrarily from the Euclidean distances (due to its topological equivalence). For our purposes, however, it will suffice to scale all eigenvalues by the same factor, and put
\begin{eqnarray}
	\alpha & := &\frac 1{\lambda_{\max}(\vec A)}\label{eqn:scaling-factor}\\
	\text{to define the scaled quadratic form }Q(\vec x) & := & \vec x^\top(\alpha\cdot \vec A)\vec x\label{eqn:quadratic-form-scaled}\\
	\text{whose induced norm satisfies } \norm{\vec x}_Q = \sqrt{Q(\vec x)} & \leq & \norm{\vec x}_2\label{eqn:euclid-bound}
\end{eqnarray}
since its largest eigenvalue was just scaled down to 1. 

We show that $\norm{\vec x}_Q$ induces distances that are consistent with the closeness-relations induced by the set of desired distances $\set{\delta_{ij}: 1\leq i<j\leq m}$ in the following sense: let $i,j,k$ be a triple selection from the set of given points, for which we want $\vec y_i$ to be closer to/equal/or farther away from $\vec y_j$ than $\vec y_k$. The possible relations can be that the distance from point $i$ to point $j$ is less than, equal to, or larger than the distance to the third point $k$, which we jointly express as 
\begin{equation}\label{eqn:proximity-relation}
	\delta_{ij} \lesseqqgtr \delta_{ik}.
\end{equation}

The norm $\sqrt{Q(x)}$ by construction is consistent with \eqref{eqn:proximity-relation} for all possible triples of neighboring points $\vec z_i,\vec z_j,\vec z_k$ (dropping the double index to simplify notation and to signify that these points are arbitrary ones from the list of existing), since
\begin{equation}\label{eqn:desired-distance-relation}
	\norm{\vec z_i-\vec z_j}_Q = \sqrt{\alpha}\cdot \delta_{ij} \lesseqqgtr  \sqrt{\alpha}\cdot \delta_{jk}=\norm{\vec z_i-\vec z_k}_Q,
\end{equation}
in which the $<$, $=$ or $>$ relations are in all cases preserved since $\alpha>0$, as $\vec A$ is positive definite. 

The Euclidean $\eps$-neighborhoods of each given data point $\vec y_i$ will be contained in the neighborhoods that our constructed norm induces, and the same $\norm{\cdot}_Q$-ball will contain points that do have the desired distances between them, up to the constant multiple $\alpha$. The proof is complete.
\end{proof}

Intuitively and informally, we can rephrase Theorem \ref{thm:embedding-any-number-of-points} as 
\begin{quote}
	Any number of given data points can be placed within disjoint neighborhoods of (virtual cluster center) points in a higher-dimensional space, whose distances we can freely control (up to a common proportionality factor).
\end{quote}

Roughly speaking, we can view Theorem \ref{thm:embedding-any-number-of-points} as somewhat dual to Theorem \ref{thm:relative-anzahl-absoluter-abstaende}: Theorem \ref{thm:relative-anzahl-absoluter-abstaende} lets us embed a \emph{relative} number (w.r.t. the dimension of the space) of points at desired \emph{absolute} distances. Conversely, Theorem \ref{thm:embedding-any-number-of-points} lets us embed an \emph{absolute} number (in the sense of being independent of the dimension) of points into a (larger) space, such that \emph{relative} proximities (in the sense of proportional distances in the larger space) can be accomplished. 

It is, however, worth discussing whether we can accomplish the desired distances exactly without a proportionality factor.
	The points $\vec z_1,\ldots,\vec z_m$ whose differences are the $\vec x_k$'s, all live in the higher-dimensional space $\R^h$, so what does this mean about the distances between the original points $\vec y_1,\ldots,\vec y_m\in\R^{\ell}$? From \eqref{eqn:euclid-bound}, we get $\norm{\vec y_i-\vec z_i}_Q\leq\norm{\vec y_i-\vec z_i}_2<\eps$, and likewise for the distance between $\vec y_j$ and $\vec z_j$. Furthermore, the distance $\norm{\vec z_i-\vec z_j}_Q=\sqrt{\alpha}\cdot\delta_{ij}$, so that we have 
	\[
	\norm{\vec z_i-\vec z_j}_Q-2\eps\leq \norm{\vec y_i-\vec y_j}_Q\leq \norm{\vec z_i-\vec z_j}_Q+2\eps,
	\]
	as Figure \ref{fig:distances} shows.
	
	\begin{figure}[h!]
		\centering
		\includegraphics{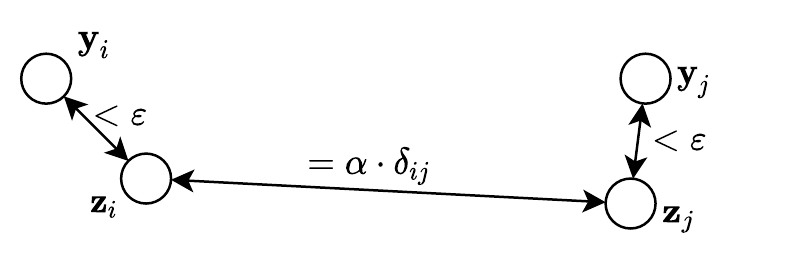}
		\caption{Distance of $\vec y_i,\vec y_j$ based on the separation of $\vec z_i$ and $\vec z_j$}\label{fig:distances}
	\end{figure}
	
	Alas, we cannot use this to claim that the distance relations between $\vec y_i, \vec y_j$ are satisfied in a likewise sense as \eqref{eqn:proximity-relation} would prescribe, because the value of $\alpha$ depends on $\eps$ through the location of the points $\vec z_i$ relative to $\vec y_i$. 
	
	Therefore, we cannot just increase the distances $\delta_{ij}$ to retain the validity of all relations \eqref{eqn:proximity-relation} even in light of the $2\eps$-error to be expected when we look at $\vec y_i,\vec y_j$.

\section{$\eps$-Semimetrics to Manipulate Distances}\label{sec:eps-semimetric}
We can carry the ideas further by avoiding the explicit work in the high-dimensional space, and instead endow the original data space with an $\eps$-semimetric that returns the desired distances. The necessity of using only an $\eps$-semimetric instead of a full metric is due to the fact that we cannot alter a norm-induced topology in $\R^{\ell}$ to an arbitrary extent, which is precluded by the equivalence of all norms on $\R^\ell$, especially any norm that we construct would be equivalent to the Euclidean norm. However, we can bypass this natural limit by a simple trick: we define a function $d:\R^{\ell}\times\R^{\ell}\to \R$ that does the following to compute the value $\d(\vec x,\vec y)$ for $\vec x,\vec y\in\R^{\ell}$:
\begin{enumerate}
	\item It takes the input points $\vec x,\vec y$ and adds noise to map them into the high-dimensional space $\R^h$. The noise will be such that it ``depends'' on both, $\vec x$ and $\vec y$, with the goal of creating linearly independent noisy $\vec z$-substitutes, which we need for the previous construction (cf. Figure \ref{fig:random-neighbors}). Let us call the resulting points $\vec z_x$ and $\vec z_y$.
	\item It then evaluates the distance $\norm{\vec z_x-\vec z_y}_Q$ inside $\R^h$, with the norm $\norm{\cdot}_Q$ as given by Theorem \ref{thm:embedding-any-number-of-points}, and returns this value as the value of $\d(\vec x,\vec y)$.
\end{enumerate}
We stress that the noise being added here is only used ``internally'' by the function $\tilde d$, but not brought onto the original data that we seek to cluster. Its purpose is merely to achieve linear independence. 

The result obtained by materializing this plan is twofold: we can show that $\d$ is an $\eps$-semimetric, for any $\eps>0$ that we fix in advance, and second, we can let the noise even be deterministic. This is important on its own right, since it means that we do not add information to the pair $(\vec x,\vec y)$ when evaluating their distance. In other words, a clustering building upon our $\eps$-semimetric will judge $\vec x$ against $\vec y$ only on grounds of the features that these two carry, but nothing else.

We make the outline above rigorous now as the proof of the following theorem:
\begin{theorem}\label{thm:semimetric}
Let a finite set of points $Y=\set{\vec y_1,\ldots,\vec y_m}\subset\R^{\ell}$ with $\ell\geq 1$ be given and let $h=\binom{m}{2}$. To each pair of distinct points $\vec y_i,\vec y_j$, assign a positive number $\delta_{ij}\leq\norm{\vec y_i-\vec y_j}_2$, only further constrained to satisfy $\delta_{ij}=\delta_{ji}>0$ for $i\neq j$. Fix any $\eps>0$.

Then, we can construct a deterministic function $\d:\R^{\ell}\times\R^{\ell}\to\R$ that is an $\eps$-semimetric that satisfies $\d(\vec y_i,\vec y_j)=\delta_{ij}$ for all pairs $\vec y_i,\vec y_j\in\R^{\ell}$.
\end{theorem}
\begin{proof}
The proof outline is essentially the two steps from above, but only made rigorous. To this end, let $\eps>0$ be given, and choose some deterministic function $f:\R^{\ell}\times\R^{\ell}\to\R$ with two properties:
\begin{eqnarray}
\abs{f}&\leq&\frac\eps 6\label{eqn:noise-bound}\\
f(\vec x,\vec y)&\neq& f(\vec y,\vec x)\text{~for all~}\vec x\neq \vec y.\label{eqn:non-symmetry}
\end{eqnarray}
The mapping of $y\in\R^{\ell}\mapsto z\in\R^h$, with $h=\binom m 2$ is done by using the function $f$ to add the ``random noise''. Note that the purpose of this random noise was, however, just to create linear independence, and we can accomplish this also on entirely deterministic ways: Lemma \ref{lem:almost-surely-rank-1} implies that each pair $(\vec y_i,\vec y_j)$ admits, with probability 1, a selection of two neighbors $\vec z_{i,j}\in B_{Q,\eps/6}(y_i)$ and $\vec z_{j,i}\in B_{Q,\eps/6}(y_j)$ such that the pairwise differences $\set{\vec z_{i,j}-\vec z_{j,i}}_{ij}$ are linearly independent. Since this selection can be done with probability 1, suitable neighbor points \emph{exist} and we can fix them as constants. The ``noise'' is then the additive term $f(\vec y_i,\vec y_j)=s\cdot\eps_{ij}$ that makes $\vec y_i'+\eps_{ij}=\vec z_{i,j}$ (likewise gives $\vec y_j'=\vec y_j+f(\vec y_i,\vec y_j)=\vec z_{j,i}$), where $\vec y_i'=(\underbrace{\vec y_i}_{\in\R^{\ell}},0,\ldots,0)\in\R^h$ and $\vec y_j'$ are the canonic embeddings of $y_i,y_j$ into $\R^h$, and with the scaling constant $s>0$ chosen to make $\abs{f(\vec y_i,\vec y_j)}=\norm{s\cdot \eps_{ij}}_Q<\eps/6$ so that the noise has the desired small magnitude. Since $f$ adds a fixed noise vector in all cases, the noise function itself is deterministic. For pairs $(\vec x,\vec y)\notin Y\times Y$, we let $f$ return a zero vector (i.e., add no noise to the points).

On the chosen neighboring points, we can define the quadratic form and induced norm that satisfies $\norm{\vec z_{i,j}-\vec z_{j,i}}_Q=\delta_{ij}$. This construction is possible, since by hypothesis, the distances desired are all less or equal to the Euclidean distances, so we do not have problems with a scaling. The $\eps$-semimetric is then defined as
\[
	\d(\vec x,\vec y) := \norm{(\vec x,f(\vec x,\vec y)) - (\vec y,f(\vec y,\vec x))}_Q.
\]
It remains to show that this is indeed an $\eps$-semimetric, so we check the corresponding properties:
\begin{enumerate}
	\item Identity of indiscernibles: since $f$ is a deterministic mapping, we get the noise $f(\vec x,\vec x)$ identically in both terms, giving $\norm{(\vec x,f(\vec x,\vec x)) - (\vec x,f(\vec x,\vec x))}_Q=0$, since we have a norm on $\R^h$. 
	\item Positivity: for $\vec x\neq \vec y$, we have $(\vec x,f(\vec x,\vec y))\neq (\vec y,f(\vec y,\vec x))$, and since $\norm{\cdot}_Q$ is a norm on $\R^h$, positivity follows directly.
	\item Symmetry: directly by construction, since 
	
	$\d(\vec x,\vec y) = \norm{(\vec x,f(\vec x,\vec y)) - (\vec y,f(\vec y,\vec x))}_Q = \norm{(\vec y,f(\vec y,\vec x)) - (\vec x,f(\vec x,\vec y))}_Q = \d(\vec y,\vec x)$.
	\item Approximate triangle inequality: for this, we first prove the following claim, saying that $\d$ and $d$ differ by at most $\frac\eps 3$ by construction: to see this, let the noise part $f(\vec x,\vec y), f(\vec y,\vec x)$ be the vectors $\bm\epsilon_1, \bm\epsilon_2$. Then (using canonic embeddings of $\vec x,\vec y\in\R^{\ell}$ inside $\R^h$), $\d(\vec x,\vec y)=\norm{(\vec x,f(\vec x,\vec y))-(\vec y,f(\vec y,\vec x))}_Q=\norm{\vec x-\vec y + (\bm\epsilon_1 - \bm\epsilon_2)}_Q$ and by the converse triangle inequality, we get
	\begin{eqnarray}
		\abs{\d(\vec x,\vec y)-\norm{\vec x-\vec y}_Q}&\leq& \norm{\d(\vec x,\vec y) - (\vec x-\vec y)}_Q\nonumber \\
		&=&\norm{(\vec x-\vec y)+(\bm\epsilon_1-\bm\epsilon_2)-(\vec x-\vec y)}_Q\nonumber\\
		&=&\norm{\bm\epsilon_1-\bm\epsilon_2}_Q\leq\norm{\bm\epsilon_1-\bm\epsilon_2}_2\leq2\frac\eps 6=\frac\eps 3\label{eqn:approximation},
	\end{eqnarray}
	using the bound $\norm{\bm\epsilon_1}_2=\norm{f(\vec x,\vec y)}_2=\abs{f(\vec x,\vec y)}\leq\eps/3$ that likewise holds for $\bm\epsilon_2$.
	Consider the metric $d(\vec x,\vec y):=\norm{\vec x-\vec y}_Q$ on $\R^{\ell}$, using the canonic embedding the points $\vec x,\vec y$ into $\R^h$. This metric satisfies the triangle inequality and is no more than $\eps/3$ different from $\d$, which will establish the claim: 
\begin{eqnarray*}
	\d(\vec x,\vec y)&\stackrel{\eqref{eqn:approximation}}\leq& d(\vec x,\vec y)+\frac\eps 3\leq d(\vec x,\vec z)+d(\vec z,\vec y)+\frac\eps 3\\
	&\stackrel{\eqref{eqn:approximation}}\leq& \d(\vec x,\vec z)+\frac\eps 3+\d(\vec z,\vec y)+\frac\eps 3+\frac\eps 3=\d(\vec x,\vec z)+\d(\vec z,\vec y)+\eps.
\end{eqnarray*}
\end{enumerate}

\end{proof}
We remark that this result, unlike Theorem \ref{thm:embedding-any-number-of-points}, is not probabilistic, as it asserts the existence of $\d$ for sure (not only with probability 1).

Our practical evaluation in the next section combined the choice of neighboring points with the definition of the noise function $f$, by seeding an available pseudorandom number generator with a value that depends non-symmetrically on both $\vec y_i,\vec y_j$ (i.e., provides different seeds for $(\vec y_i,\vec y_j)$ than for $(\vec y_j,\vec y_i)$), and which produced a pseudorandom uniform distribution in $[0,1]$. This delivered the desired pseudorandom (and hence deterministic) neighbors $\vec z_{i,j}$ and $\vec z_{j,i}$, while achieving linear independence with high probability. Scaling the pseudorandom noise down by a factor $0<s<\frac\eps{6\cdot\sqrt{h}}$ accomplished bound \eqref{eqn:noise-bound}.

It is worth remarking that $\d$ cannot in general satisfy the triangle inequality with $\eps=0$, as a counter-example situation shown in Figure \ref{fig:triangle-inequality-counterexample} shows.

\begin{figure}
	\centering
	\includegraphics[scale=0.7]{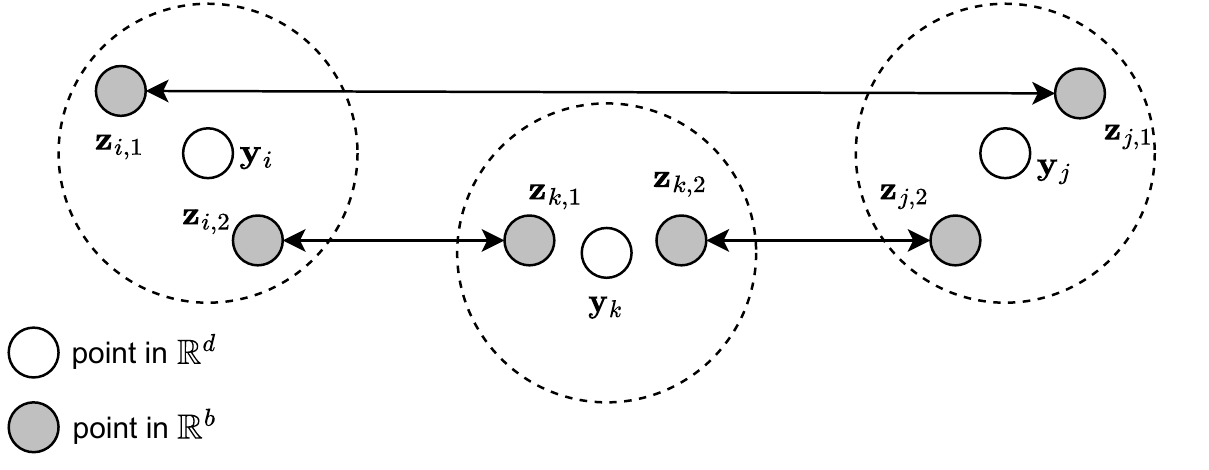}
	\caption{Violation of the triangle inequality: the triple $\vec y_i,\vec y_k,\vec y_j$ would satisfy the triangle inequality on the distances between them. However, the distances between $(\vec z_{i,2},\vec z_{k,1})$ and between $(\vec z_{k,2},\vec z_{j,2})$ add up to a value less than the direct distance from $\vec z_{i,1}$ to $\vec z_{j,1}$. Hence, the triangle inequality cannot generally hold for distances between $\vec y_i,\vec y_j,\vec y_k$ evaluated on the $\vec z$-neighbors as done by the $\eps$-semimetric $\d$, since the neighbors are determined under other constraints than this inequality. }\label{fig:triangle-inequality-counterexample}
\end{figure}


\section{Experimental Demonstration}

We implemented the construction of Theorem \ref{thm:embedding-any-number-of-points} in Octave \cite{octave} as a script that does the following\footnote{All scripts mentioned in this article will be made publicly available, if the paper receives positive reviews}:
\begin{itemize}
	\item We randomly drew points in $\R^\ell$ from a standard Gaussian distribution (with zero mean and unit variance), corresponding to hypothetical data points with (only) $\ell$ features, mainly for plotting. For the following experiment, we had $m=10$ and $\ell=2$, for plotting and to avoid overly lengthy tables.
	\item We iterated pairwise through the $h=\binom m 2$ distinct pairs of data points, assigning another random yet positive distance uniformly chosen as $\delta\sim\mathcal{U}[1,3]$.
	\item Then, we construct the metric $q$ as described above, and verified whether the desired distances come up as we randomly chose them. Table \ref{tbl:distances-designed} shows the scaled distances versus the values of the quadratic form $q$, thus experimentally confirming that the distances come up as we wanted them. The embedding into $\R^h=\R^{45}$ is done by adding a random noise vector with $43$ dimensions, scaled in magnitude by a factor of $\eps/\sqrt{h}$ with $\eps=0.1$, to accomplish an Euclidean distance of the $\vec z$-neighbor to the point $\vec y$ within $\leq\eps$; cf. Figure \ref{fig:random-neighbors}.
\end{itemize}

\begin{table}[h!]
	\caption{Experimental verification of proportional pairwise distances, randomly assigned pairwise between 10 points in $\R^2$.}\label{tbl:distances-designed}
	\includegraphics[width=\textwidth]{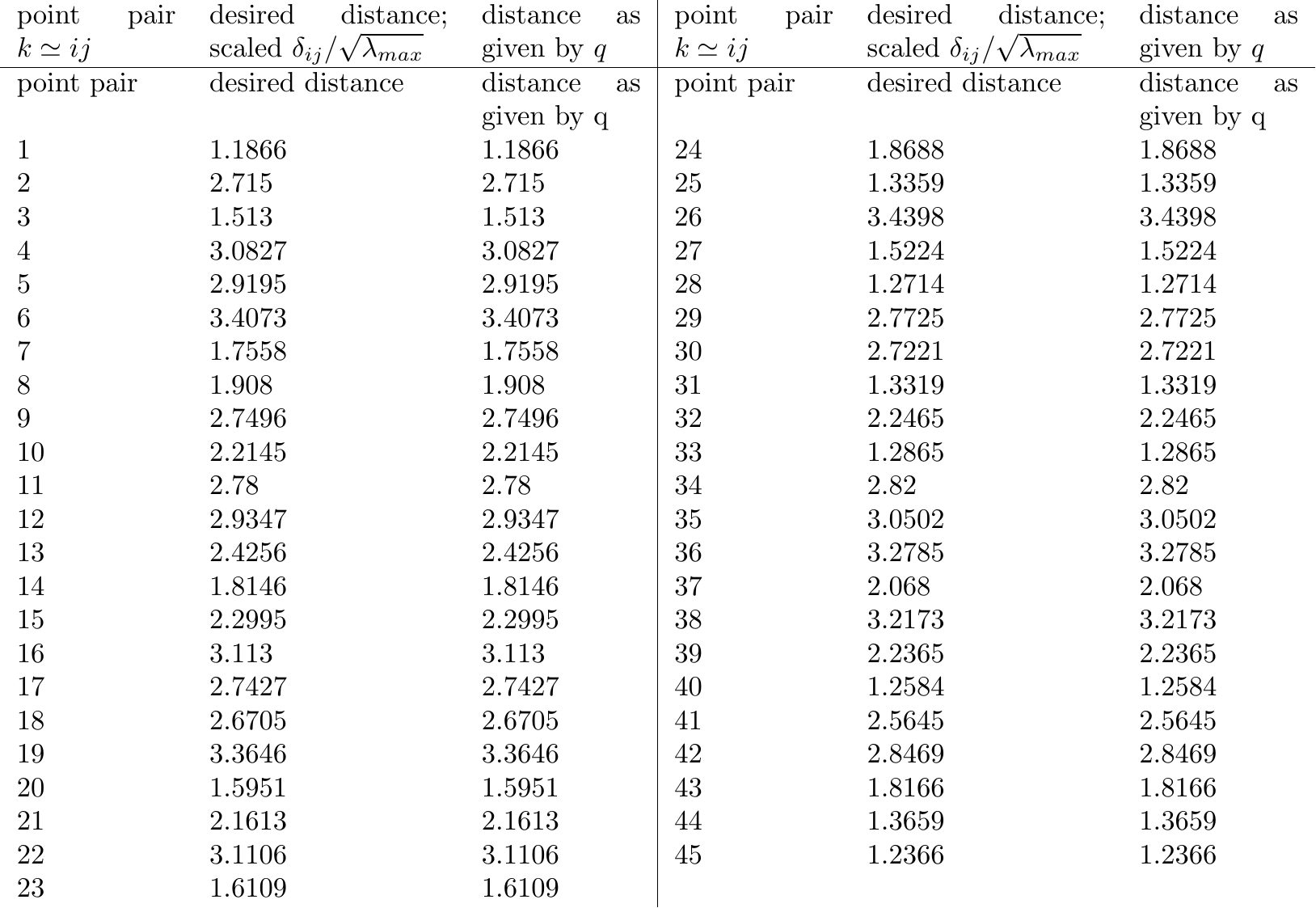}
\end{table}

\begin{remark}[On the Reproducibility of the results given here]
Conditional on the random number generator to remain as implemented in Octave as of version 7.1.0 that we used \cite{octave}, the results in this work are reproducible by seeding the random number generators with the value $12345$ or $456789$. We put the respective line of code upfront in the scripts, and recommend to include it or remove it, to either reproduce the results given here, or run one's own verification experiments.
\end{remark}

\subsection{Attacking Clustering: General Outline}\label{sec:first-experiment}
We implemented the above construction in Octave \cite{octave}, with the following basic setup: the number $m$ and dimension $\ell$ of points can be chosen arbitrarily at the beginning. The program proceeds by computing all pairwise distances, and assigning random values $\delta_{ij}\sim \mathcal{U}[1,3]$ to the distances. It then proceeds by enlarging the dimension from $\ell$ to $h$ and padding the additional coordinates with yet more independent random values, drawn from a Gaussian distribution such that the $5\sigma$-range falls within an interval of $(-\eps,+\eps)$, with $\eps$ being set to $\eps=0.1$ at the beginning (but changeable). Lemma \ref{lem:almost-surely-rank-1} is then verified by computing the rank of the resulting matrix $\vec M$ and checking it to be maximal. 

After that, the code constructs the quadratic forms using singular value decomposition, adds up the resulting matrices and scales them down by the largest eigenvalue of the total. The final check is done by computing the distances according to the metric $Q$ as defined by \eqref{eqn:quadratic-form-scaled}, and verifying its value against the (scaled) random distances $\delta_{ij}/\lambda_{\max}$, with $\delta_{ij}$ as chosen at the beginning. The equality of the desired and computed distances are, repeatably, satisfied, thus verifying the correctness of the construction and both, Theorems \ref{thm:relative-anzahl-absoluter-abstaende} and Theorem \ref{thm:embedding-any-number-of-points}, experimentally. Enlarging the dimensions or number of points, however, quickly shows that the construction scales badly, since the dimension of the larger space grows quadratically with the number of points. Thus, the construction is practically feasible at best for a small number of points and hence clusters.

An adversarial use of Theorem \ref{thm:embedding-any-number-of-points}, yet up to experimental validation in follow-up work, to forge a clustering may then proceed along the following recipe: 
\begin{enumerate}
	\item Fix a value $\eps>0$ such that all points are pairwise more than $\eps$ apart in the Euclidean distance (the computation of $\eps$ is a closest-pair of points search problem).
	\item Inside the space $\R^h$ with $h=\binom m 2$, associate $m-1$ random points $\vec z_{i,s}\in\R^h$ for $s=1,2,\ldots,i-1,i+1,\ldots,m$ to each $\vec y_i\in\R^{\ell}\subset\R^h$ within an Euclidean distance $\norm{\vec z_{i,s}-\vec y_i}_2<\eps$, for $i=1,2,\ldots,m$.
	\item For each pair $\vec y_i,\vec y_j$ of distinct vectors (assuming $i<j$), pick a distance value $\delta_{ij}>0$, and assign this distance to a pair of representatives $\vec z_{i,a}, \vec z_{j,b}$ for some neighbor (indexed as $a$) of $\vec y_i$ and another neighbor (indexed as $b$) of $\vec y_j$. Mark these two neighbors as ``used'', and for the next distance, choose the neighbors afresh.
	\item Construct the quadratic form $q$ as prescribed in Sect. \ref{sec:embedding points at}, and by scaling with the factor $\alpha$ (eq. \eqref{eqn:scaling-factor}), define the quadratic form $Q(\vec x)$ (eq. \eqref{eqn:quadratic-form-scaled}) and induced norm $\norm{\vec x}_Q$ (eq. \eqref{eqn:euclid-bound}).
	
	A practical difficulty comes up with the dependence of $\alpha$ on both, $\eps$ and the set of chosen neighbors $\vec z$ around the given data points. That is, we cannot choose the distances among the $\vec z$-vectors independently of $\eps$, since these distances will scale by a factor that indirectly depends on $\eps$ and is outside our full control. In that sense, the parameterization of the algorithm needs care. We will come back to the parameterization in Section \ref{sec:k-means-experiment} and Section \ref{sec:dbscan}, when we describe how we chose the parameters to drive the algorithms towards the desired results.
	
	\item Run the clustering, telling the algorithms to, e.g., use some chosen points as cluster centers, assigning all points within $\eps$-distance, measured by the metric that $\norm{\cdot}_Q$ induces, to the respective cluster. By designing the distances between the cluster centers accordingly to be less than $\eps$ or (much) larger than $\eps$, one can either merge several cluster centers, or keep them separated, just as desired.

	Since this is a one-to-one association of virtual cluster centers to points, one may need parameterize the algorithms accordingly, or put the points $\vec z_{i,j}$ into close neighborhood of one another (for distinct $i$) or far away, such that the algorithm returns the desired results. We will demonstrate this for DBSCAN and $k$-Means.
\end{enumerate}

The experimental verification of Theorem \ref{thm:semimetric} is done along using its constructed $\eps$-semimetric to manipulate standard clustering algorithms. The next sections are devoted to a demonstration of this. In addition, the verification of Theorem \ref{thm:semimetric} went successful along the same lines as outlined above for Theorem \ref{thm:relative-anzahl-absoluter-abstaende} and Theorem \ref{thm:embedding-any-number-of-points}, with the additional verification of the triangle inequality to hold up to the additive $\eps$-error. This was successfully verified in another Octave script.

\subsection{Manipulating $k$-Means}\label{sec:k-means-experiment}
We tested the construction towards manipulating the well-known $k$-Means clustering. The experiment followed these steps, all implemented in an Octave script:
\begin{enumerate}
	\item Generation of a set of $n$ points $\vec y_1,\ldots,\vec y_m\in\R^{\ell}$ uniformly at random (as in Section \ref{sec:first-experiment}).
	\item Run through the list $\vec y_1,\ldots,\vec y_m$ and to each $y_i$, assign a random but fixed class $class(\vec y_i)\in\mathcal{U}(\set{1,2,3})$, \emph{irrespectively} of any features of $\vec y_i$ (i.e., independent of its location relative to other points).
	
	Figure \ref{fig:point-cloud-example} shows an example of the inputs to the algorithm as constructed in this step, already assigning the desired class (here chosen at random). To visualize the input the plot of point cloud uses different symbols to represent the desired classes. Visually this confirms that the classes, since chosen at random, have nothing to do with any geometric proximity (at least if it were the Euclidean distance).
	
	\begin{figure}[h!]
		\centering
		\includegraphics{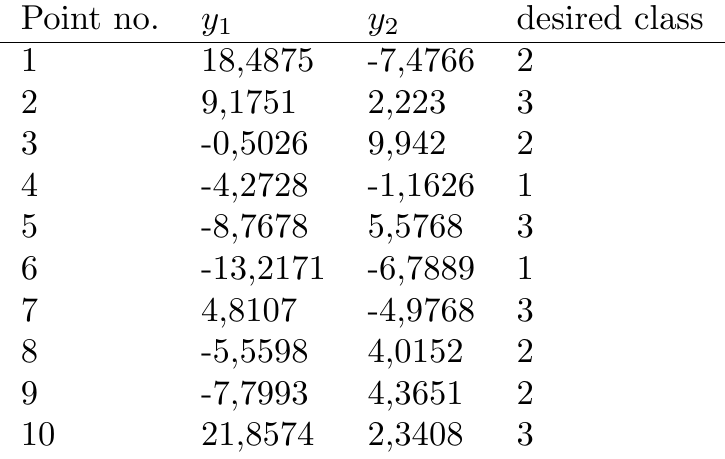}
		
		\includegraphics[width=0.6\textwidth]{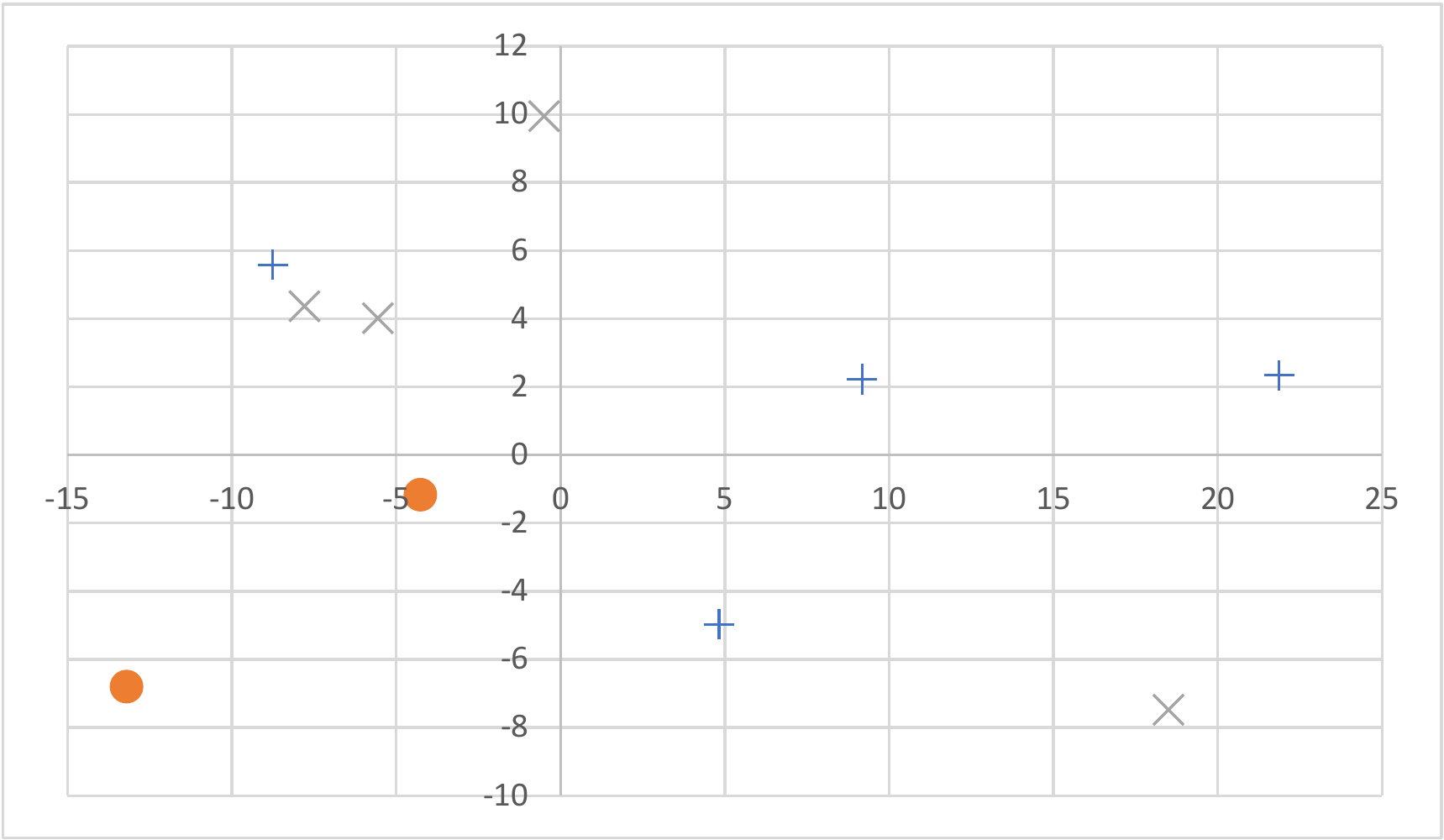}
		\caption{10 example points with spatial locations, and markers representing the different \emph{desired} classes}\label{fig:point-cloud-example}
	\end{figure}
	
	\item To assign the distances between points of the same cluster, versus distances between points in distinct clusters, we determined the closest pair of points among the set $Y$ at distance $\delta_0=\min_{i,j}\norm{\vec y_i-\vec y_j}_2$. We likewise computed the largest separation as $\delta_1 = \max_{i,j}\norm{\vec y_i-\vec y_j}_2$. Based on $\delta_0,\delta_1$, we defined the ``small distance'' to be $\delta_{small}=\frac 1{200}\cdot\delta_0$, and a ``large distance'' to be $ \delta_{large}=200\cdot\delta_1$. 
	\item Iterating over all pairs $\vec y_i,\vec y_j$ (with $i\neq j$), we assigned the desired distances between them as follows:
	\[
		\delta_{ij}=\begin{cases}
			\delta_{small} & \text{if~}class(\vec y_i)=class(\vec y_j),\\
			\delta_{large} & \text{if~}class(\vec y_i)\neq class(\vec y_j)
		\end{cases}
	\]
	Using these values $\delta_{ij}$, we constructed the norm $\norm{\cdot}_Q$ in the (large-dimensional) space $\R^h$ as described in Section \ref{sec:embedding points at} and Section \ref{sec:generalized-construction}, and defined the $\eps$-semimetric as $\d(x,y):=\norm{(\vec x, PRG(\vec x,\vec y))-(\vec y,PRG(\vec y,\vec x))}_Q$, where $PRG(seed)$ is Octave's built-in pseudorandom number generator with the seed value $seed$ (computed from both, $\vec x=(x_1,\ldots,x_\ell)$ and $\vec y=(y_1,\ldots,y_\ell$) as $seed(\vec x,\vec y)=0.9\cdot\frac{\eps}{\sqrt{h}}\cdot\left(\prod_i x_i\right)\cdot\left(\sum_j y_j\right)$. This formula was chosen \emph{arbitrarily}, and can be replaced by any (other) formula to produce a seed for $(\vec x,\vec y)$ differently to a seed produced from $(\vec y,\vec x)$, only subject to satisfy the properties \eqref{eqn:noise-bound}, \eqref{eqn:non-symmetry} that the proof of Theorem \ref{thm:semimetric} requires.
	
	The verification of the quadratic form $q$ giving the desired distances works exactly is shown in Table \ref{tbl:distances-designed}, and we refrain from repeating another table showing similar results.
\end{enumerate}

To verify our setup, we evaluated the distances inside a cluster, as well as the distances between clusters. Table \ref{tbl:cluster-separations} shows a snapshot of one (out of many runs), confirming that the distances are small inside a cluster and large between any two clusters. To avoid very small and very large numbers causing numeric issues, we used the unscaled version (i.e., without the downscaling by the largest eigenvalue as in \eqref{eqn:scaling-factor}) of the quadratic form here, which does no change to the desired relations of which points are close to one another and which are far away from each other. The experiments nicely show how the clusters are separated by a large distance at least (proportional to $\delta_{large}$), while points within the same cluster are separated by a small distance at most (proportional to $\delta_{small}$).

\begin{table}[h!]
	\centering
	\caption{Separation within and between clusters}\label{tbl:cluster-separations}
	\subfloat[Cluster centers ($=$ averages of all desired cluster points)]{\label{tbl:cluster-centers}
		\includegraphics{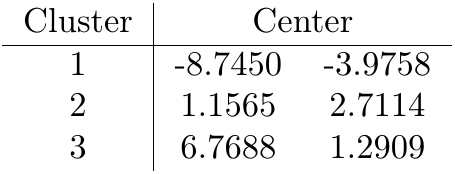}
	}

	\subfloat[Maximum distance within a cluster]{
		\begin{tabular}{lccc}
			Cluster & 1 & 2 & 3 \\ \hline
			distance within & 0.0077013 & 0.0079287 & 0.008057 \\ 
		\end{tabular}}
	
	\subfloat[Minimum distance between two clusters]{
		\begin{tabular}{l|ccc}
			Cluster & 1 & 2 & 3 \\\hline 
		1 &  & 7248.6 & 7248.6 \\ 
		2 & 7248.6 &  & 7248.6 \\ 
		3 & 7248.6 & 7248.6 &  \\ 	
		\end{tabular}
	}
\end{table}

With this setup, we ran a slightly adapted version of the $k$-Means implementation provided by \cite{ghatakDBSCAN2017}. The change we made was twofold, namely:
\begin{itemize}
	\item We adapted the function to take our $\eps$-semimetric to compute distances,
	\item and we kept the number of iterations in the algorithm fixed by not letting it terminate earlier than before reaching the maximum number of iterations.
\end{itemize}

The $k$-Means algorithm then received the following input: we computed the desired cluster centers by grouping and averaging the points $\vec y_i$ according to their cluster. The cluster centers computed (also reproducible from the table in Figure \ref{fig:point-cloud-example}), are shown in Table \ref{tbl:cluster-centers}.

That is, the cluster center for the class $c$ is the arithmetic mean of all points $\vec y_i$ to which we assigned $class(\vec y_i)=c$, and this was done for all classes. The resulting cluster centers were then included in the data points (see Table \ref{tbl:clustering-results}, columns to the left of the vertical separator line) on which we constructed the metrics following the setup steps above, and put in front of the list of input points to the $k$-Means algorithm. This covers two versions of how $k$-Means can be applied: we can either directly supply it with the cluster centers (this is what our implementation does), and let it run on the given data points (which then include the cluster centers accordingly). Alternatively, we can refrain from supplying cluster centers, in which case $k$-Means (in the given implementation) would take the first $k$ points as cluster centers to start with if there are $k$ classes. Since we have put the cluster centers upfront in the list, the algorithm will in both cases be initialized identically. 

\paragraph{Results:} We found that the construction is numerically relatively unstable, since the eigenvalues of the matrix $\vec A$ to construct the quadratic form in some cases fell out of the numerical scope of the machine precision. Whenever the experiment completed without over- or underflows, however, the clustering of $k$-Means came up \emph{exactly as we desired}. Table \ref{tbl:clustering-results} shows the results of $k$-Means in the column ``by $k$-Means''. It is directly visible that $k$-Means produced just the random classes assigned. The experiment is repeatable with a freshly seeded random number generator each time, and confirmed the findings also over many independent repetitions. We also let the algorithm run on cluster centers that were simply chosen as arbitrary $z$-neighbors of some point within the desired cluster, but not the direct arithmetic mean of the cluster as such. This made the algorithm frequently update (shift) the cluster center, which due to the construction of our quadratic form, led to strong changes in the distances as were computed. The resulting cluster assignments were, in that case, far off what we desired them to be. The choice of the arithmetic mean of the desired points to be the cluster centers, on the contrary, will effectively let the algorithm re-compute the same cluster center over and over again, so that the distances that we designed really come up as desired. This leads to the results reported above.

\subsection{Manipulating DBSCAN}\label{sec:dbscan}
With the same setup as for $k$-Means, we let DBSCAN do a clustering for us, with the following differences to section \ref{sec:k-means-experiment}:
\begin{itemize}
	\item No inclusion of pre-computed cluster centers; the data was given to DBSCAN just as it came out of the random generators
	\item The neighborhood size $\epsilon$ and minimum number $minPts$ of points within the $\epsilon$-neighborhood was set to $\epsilon = \frac{\delta_{small}+\delta_{large}}2$ and $minPts=2$. 
\end{itemize}

\paragraph{Results:} As far as the scalability issue concerns the construction of the norm $\norm{\cdot}_Q$ on $\R^h$, this issue exists in all experiments and relates to the method in general. However, unlike $k$-Means, the experiments on DBSCAN ran without any numerical issues and always terminated, in all cases delivering an ``isomorphic'' clustering. That is, DBSCAN used a different naming (enumeration) of the classes, but there was always a one-to-one correspondence between the class $c$ that DBSCAN assigned to the point $y_i$, and the (randomly desired) class $class(y_i)$. Table \ref{tbl:clustering-results} shows the results in the column ``by DBSCAN'', where it is visible that the cluster naming is not matching the original clusters, but remains ``the same'' via the one-to-one correspondence $\set{1\mapsto 3, 2\mapsto 1, 3\mapsto 2}$, between the randomly assigned desired classes and those that DBSCAN recovered.

\begin{table}[h!]
	\caption{Experimental verification of proportional pairwise distances, randomly assigned pairwise between 10 points in $\R^2$. The right columns show what $k$-Means and DBSCAN recovered using the input data on the left, and our designed $\eps$-semimetric to measure distances.}\label{tbl:clustering-results}
	\includegraphics[scale=1]{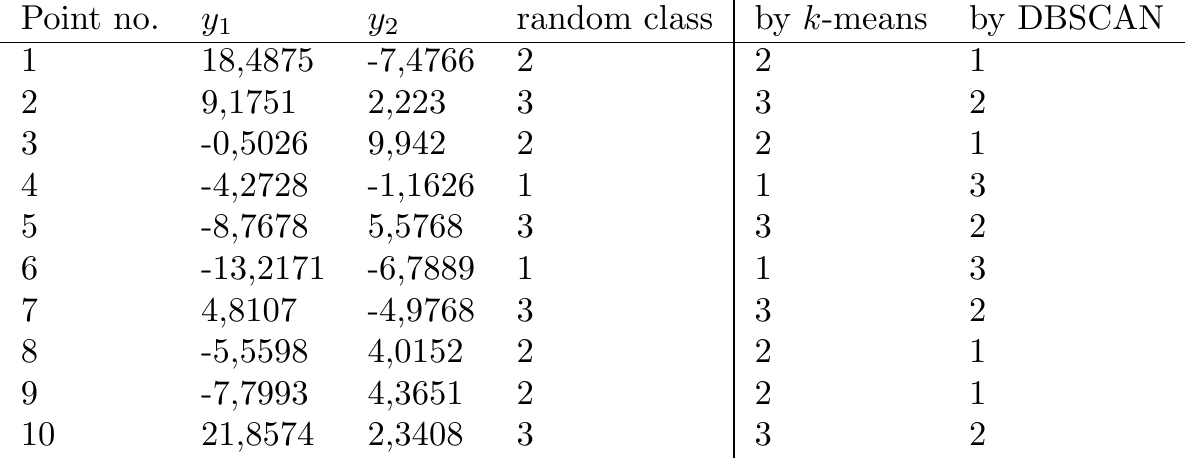}
\end{table}

\section{Related Work}
Notable in relation to our work are methods to design norms (and hence distance functions), detectors or other methods to become robust against noise, such as done in \cite{wuNoiseClusteringUsing2006,gaoLLpMetricBased2009,zhangRobustEmbeddedDeep2019,shahFrequencyCentricDefense2021,chenSuperResolutionCodingDefense2020}, or to account for, e.g., correlations in the data to define clusters \cite{beerLUCKLinearCorrelation2019}. In light of this, observe that our choice of the Euclidean norm is indeed itself arbitrary, and given the equivalence of all norms, more robust choices are compatible with our constructions. In particular, since we can allow the noise magnitude to be arbitrarily small towards imperceptibility \cite{aydinImperceptibleAdversarialExamples2021} (as long as the resulting constructions remain numerically feasible), our methods can still work with noise-robust norms instead of the Euclidean norm as we used. It amounts to a change of multiplicative constants. Other work that tackles the problem of cluster center determination (such as we require for $k$-Means for example), is as well compatible with our constructions; noting that the cluster centers can be computed by any means, including sophisticated methods as in \cite{khanClusterCenterInitialization2004,saabEstimationClusterCentroids2016}, for example. In our case, whether the clustering algorithm will work well with the user-supplied cluster centers depends on how much the cluster centers may become adapted by the algorithm. For $k$-Means, as we tested above, letting the algorithm change the cluster centers destroys the designed effects and outcomes (see the experiments above). 

A notable non-adversarial instance of the idea to replace the metrics used for clustering was also reported in \cite{daiMetricbasedGenerativeAdversarial2017}, which designed the metric for separating real from adversarial examples to train generative adversarial networks (GANs). An idea similar to our deterministic addition of noise has also been used to improve the quality of clustering, such as done in \cite{zhangDensityBasedMultiscaleAnalysis2018}, where the number of neighbors in increasing distance radii is used to map points into a high-dimensional feature space.	

Many studies of clustering \cite{deldjooSurveyAdversarialRecommender2021,daiMetricbasedGenerativeAdversarial2017,zhangDensityBasedMultiscaleAnalysis2018,xieStatisticallyRobustClusteringTechniques2022,biggioDataClusteringAdversarial2013,aydinImperceptibleAdversarialExamples2021,zhangRobustEmbeddedDeep2019} are concerned with noise robustness or poisoning attacks (e.g., \cite{liuUsingSingleStepAdversarial2021}), and demonstrate how sensitive the algorithm may become upon insertion of (even a few) adversarial examples. Our work extends these thoughts to not only poisoning the input to the algorithm, but also its configuration by replacing the distance functions a posteriori to justify some (desired) behavior. This is particularly problematic in, for example, recommender systems, as was thoroughly studied in \cite{deldjooSurveyAdversarialRecommender2021}. Attacks similar to our setting were also reported as \emph{adversarial backdoors} such as proposed by \cite{zhangAdvDoorAdversarialBackdoor2021}, but with our modification being not in the algorithm, but merely in its configuration.

An interesting possible countermeasure to our attack technique are methods of posterior cluster validation \cite{jose-garciaSurveyClusterValidity2021,krishnamoorthyNewInterCluster2013}, which may judge the ``plausibility'' of clusters on different means than our crafted metric. This is indeed conceptually close to the countermeasure that we propose as a prior commitment to the exact configuration of the clustering algorithm, including the choice of distance in particular. Such commitments would be up to independent party verification, and the cited reference offers an additional tool for such third-party verification to be done. Specifically some robust classification techniques, e.g., \cite{wangRobustClusteringSocial2013}, work with specially designed metrics that should not be silently replaceable.

Some algorithm implementations (see, e.g., \cite{mcinnesComparingPythonClustering2016}) allow to specify and supply a dissimilarity matrix, specifying the $ij$-th element as the desired distance between data point $\vec y_i$ and data point $\vec y_j$, not necessarily constraining this dissimilarity matrix to correspond to any topological metric. A famous related result is Schoenberg's criterion \cite{schoenbergRemarksMauriceFrechet1935}, which gives conditions on a matrix to correspond to Euclidean distances of points placed in some (possibly high-dimensional) space. Constructive proofs of Schoenberg's criterion take the norm as fixed (to be Euclidean), and desired distances, and from this determine the placement of points such that they are in consistent locations. Our work uses the same three items, but two different ingredients for the third item as outcome: we first fix the distances, and place the points in the space, and from there, construct a metric that puts the points into the desired distance from each other (although their locations are fixed).

\section{Discussion}
Our results illustrate a systematic threat to arise if a data processor is not open on the metrics and algorithmic details of the applied clustering algorithm. Suppose that the data processor would keep the algorithm details, including its configuration parameters, as a business secret, then the following scenario would instantly become possible: two customers, both with very similar, but not entirely identical, features show up and ask for a credit loan. A similarity-based clustering like $k$-Means or DBSCAN would be expected to classify both persons in the same way, meaning that both would or both would not get the credit loan. 

However, if the provider is malicious and decides to subjectively judge yet pretend to be objective by ``leveraging artificial intelligence'', it can craft the metric $\norm{\cdot}_Q$ in a way to put customer \#1 into the rejection and customer \#2 into the acceptance bin for the credit decision, irrespectively of the fact that they have almost the same properties. Conversely, the malicious provider can also include a desired customer with entirely different properties in either the ``accept'' or ``reject'' bin, by crafting the metric accordingly.

Upon accusations from third parties, it is then even possible to claim that the function $\d$ as constructed to discriminate the customers is indeed a metric in the mathematical sense. It only differs from a true metric by a violation of the triangle inequality up to an additive error of $\eps$. However, making $\eps$ smaller than the machine precision will make $\d$ practically indistinguishable from a real metric, at least on the computer where the clustering shall be demonstrated practically (e.g., to settle a court argument).

Although experimentally verifiable, the procedure does not scale well as the experiments with $k$-Means have shown. It is imaginable to implement the whole scheme with arbitrary precision arithmetic using suitable libraries\footnote{Wikipedia provides a comprehensive list at \url{https://en.wikipedia.org/wiki/List_of_arbitrary-precision_arithmetic_software}}. The algorithms used here were selected for their simplicity and availability as ``bare bones implementation'', void of complicated overheads for integration in larger libraries. This choice was aimed at easing matters of verification and reproduction of results in this work. Future studies may thus try to use off-the-shelf implementations of the same or more powerful clustering algorithms in other programming languages such as Python or similar. This work is intended as a mere first step. Extending the experiments to more advanced clustering algorithms than the two basic ones that we have used here, is a further direction of future work. 

\section{Conclusion}
The lesson we learned from our study and experiments is the inevitable need to enforce transparency and an a priori commitment to algorithms, parameters, metrics and all details of a clustering algorithm before applying it to data. Our ``attacks'' depend on the possibility of an a posteriori manipulation of the algorithms to justify a desired result. This possibility can include reproducing a past set of legitimate clustering decisions and may only affect a manipulation of a single or a few data points (simply set the distances accordingly to match an honest setup, and only set some values to distances of malicious desire).

The operational principles to derive from these insights are well known and standard in cryptography, known as Kerckhoffs' principle, and elsewhere also called a ``nothing-up-my-sleeve'' design choice\footnote{This term stems from the field of symmetric cryptography, where constants in algorithms are chosen in a way that allows to re-calculate the constant (e.g., use a hexadecimal representation of $\sqrt 2$ as is done in some hashing algorithms or similar).}. The idea is to explain where design choices, even if they refer to seemingly arbitrary constants, come from, and are not due to some hidden agenda of the designer (``nothing up my sleeve'').

The principle is the same here: the implementer and user of the clustering algorithm should be obliged to explain all design choices and transparently and publicly commit to them before letting the algorithms become operational. If not, then a posteriori manipulations, just as we have demonstrated, become possible.


\bibliographystyle{abbrv}

\end{document}